\documentclass[11pt]{article}
\usepackage{graphicx,color}
\usepackage{amssymb}
\usepackage{latexsym}
\usepackage{epsfig,enumerate,amsmath,amsfonts,amssymb,setspace}
\usepackage{indentfirst}
\usepackage{epsfig,float}
\usepackage{wrapfig,lipsum}
\usepackage{mathrsfs}
\usepackage{times}
\usepackage{latexsym}
\usepackage{indentfirst}
\usepackage{wrapfig,lipsum}
\usepackage{amsthm}
\usepackage[T1]{fontenc}
\usepackage{authblk}
\usepackage[ruled,vlined]{algorithm2e}
  \baselineskip 8 mm

\newtheorem{theorem}{Theorem}[section]
\newtheorem{lemma}[theorem]{Lemma}

\newtheorem{claim}[theorem]{Claim}
\newtheorem{definition}[theorem]{Definition}

\title{Algorithmic aspects of disjunctive domination in graphs}
\author[1]{B. S. Panda\thanks{bspanda@maths.iitd.ac.in}}
\author[1]{Arti Pandey\thanks{artipandey2305@gmail.com}}
\author[2]{S. Paul\thanks{paulsubhabrata@gmail.com}}
\affil[1]{Department of Mathematics, Indian Institute of Technology Delhi\newline Hauz Khas, New Delhi 110016, INDIA}
\affil[2]{Advanced Computing and Microelectronics Unit, Indian Statistical Institute\newline Kolkata 700108, INDIA}

\begin{document}
\maketitle
\begin{abstract}
For a graph $G=(V,E)$, a set $D\subseteq V$ is called a \emph{disjunctive dominating set} of $G$ if for every vertex $v\in V\setminus D$, $v$ is either adjacent to a vertex of $D$ or has at least two vertices in $D$ at distance $2$ from it. The cardinality of a minimum disjunctive dominating set of $G$ is called the \emph{disjunctive domination number} of graph $G$, and is denoted by $\gamma_{2}^{d}(G)$. The \textsc{Minimum Disjunctive Domination Problem} (MDDP) is to find a disjunctive dominating set of cardinality $\gamma_{2}^{d}(G)$. Given a positive integer $k$ and a graph $G$, the \textsc{Disjunctive Domination Decision Problem} (DDDP) is to decide whether $G$ has a disjunctive dominating set of cardinality at most $k$. In this article, we first propose a linear time algorithm for MDDP in proper interval graphs. Next we tighten the NP-completeness of DDDP by showing that it remains NP-complete even in chordal graphs. We also propose a $(\ln(\Delta^{2}+\Delta+2)+1)$-approximation algorithm for MDDP in general graphs and prove that MDDP can not be approximated within $(1-\epsilon) \ln(|V|)$ for any $\epsilon>0$ unless NP $\subseteq$ DTIME$(|V|^{O(\log \log |V|)})$. Finally, we show that MDDP is APX-complete for bipartite graphs with maximum degree $3$.
\end{abstract}
\vspace*{.2cm}

\section{Introduction}
\label{sec:1}
Let $G=(V,E)$ be a graph. For a vertex $v\in V$, let $N_G(v)=\{u\in V |uv\in E\}$ and $N_G[v]=N_G(v)\cup \{v\}$ denote the open neighborhood and the closed neighborhood of $v$, respectively. For two distinct vertices $u,v\in V$, the distance $dist_G(u, v)$ between $u$ and $v$ is the length of a shortest path between $u$ and $v$. A vertex $u$ dominates $v$ if either $u=v$ or $u$ is adjacent to $v$. A set $D\subseteq V$ is called a \emph{dominating set} of  $G=(V,E)$ if each $v \in V$ is dominated by a vertex in $D$, that is, $|N_G[v] \cap D| \geq 1$ for all $v \in V$. The domination number of a graph $G$, denoted by $\gamma(G)$, is the minimum cardinality of a dominating set of $G$. For a graph $G$, the \textsc{Minimum Domination} problem is to find a dominating set of cardinality $\gamma(G)$. Domination in graphs is one of the classical problems in graph theory and it has been well studied form theoretical as well as algorithmic point of view \cite{Haynes1,Haynes2}. Over the years, many variants of domination problem have been studied in the literature due to its application in different fields
 varying from computer science to electrical engineering, operation research to network securities etc.
 The concept of \emph{disjunctive domination} is a recent and an interesting variation of domination \cite{goddard}.

In domination problem, our goal is to place minimum number of sentinels at some vertices of the graph so that all the remaining vertices are adjacent to at least one sentinel. In practice, depending upon the monitoring power, we can have different types of sentinels. To secure the graph with different types of sentinels, we need concept of different variants of domination.
Efforts made in this direction have given rise to different types of domination, such as, distance domination, exponential domination, secondary domination.
%If a sentinel can secure a vertex within a distance of $k$, then \emph{distance $k$-dominating set} is enough to secure the graph (see Chapter 12 of \cite{Haynes2}). A set $D_k\subseteq V$ is called a \emph{distance k-dominating set} of $G$ if every vertex $u\in V\setminus D_k$ is within a distance of $k$ from some vertex $v\in D_k$, i.e., $dist_G(u,v)\leq k$. Hedetniemi et al. \cite{Hedetniemi} have introduced another concept of domination, known as \emph{secondary domination}. A set $D'$ of vertices is called a \emph{$(1,k)$-dominating set} \cite{Hedetniemi} if for every vertex $u\in V\setminus D'$, there are two distinct vertices $v,w\in D'$ such that $u$ is dominated by $v$ and $w$ is within a distance of $k$ from $v$. The main feature of this type of domination is the existence of a `secondary' vertex within a distance of $k$.
%
In some cases, it might happen that the monitoring power of a sentinel is inversely proportional to the distance, that is, the domination power of a vertex reduces as the distance increases. Motivated by this idea,
%Dankelmann et al. \cite{dankelmann} have introduced the concept of \emph{exponential domination}. Recently,
Goddard et al. \cite{goddard} have introduced the concept of \emph{disjunctive domination} which captures the notion of decay in domination with increasing distance. A set $D_d\subseteq V$ is called a \emph{$b$-disjunctive dominating set} of $G$ if every vertex $v\in V\setminus D_d$ is either adjacent to a vertex in $D_d$ or there are at least $b$ vertices of $D_d$ within a distance of two from $v$. The minimum cardinality of a $b$-disjunctive dominating set of $G$ is called the \emph{$b$-disjunctive domination number} and it is denoted by $\gamma_b^d(G)$. A vertex $v$ is said
to be \emph{$b$-disjunctively dominated} by $D_{d}\subseteq V$ if either $v\in D_{d}$ or $v$ is adjacent to a vertex of $D_{d}$
or has at least $b$ vertices in $D_{d}$ at distance $2$ from it. Note that disjunctive domination is more general concept than distance two domination, since the parameter $\gamma_1^d(G)$ is the distance two domination number. For simplicity, $2$-disjunctive domination is called disjunctive domination. The disjunctive domination problem and its decision version are defined as follows:\\

\noindent\underline{\textsc{Minimum Disjunctive Domination Problem (MDDP)}}
\begin{description}
  \item[Instance:] A graph $G=(V,E)$.
  \item[Solution:] A disjunctive dominating set $D_d$ of $G$.
  \item[Measure:] Cardinality of the set $D_d$.
\end{description}
\noindent\underline{\textsc{Disjunctive Domination Decision Problem (DDDP)}}
\begin{description}
  \item[Instance:] A graph $G=(V,E)$ and a positive integer $k\leq |V|$.
  \item[Question:] Does there exist a disjunctive dominating set $D_d$ of $G$ such that $|D_d|\leq k$?
\end{description}

The concept of disjunctive domination has been introduced recently in 2014 \cite{goddard}
%. Unlike other variations of domination, it
%has not been well studied until now. As per our knowledge, only two papers have appeared on disjunctive domination
and further studied in \cite{Henning2014}. In \cite{goddard}, Goddard et al. have proven bounds on disjunctive domination number for specially regular graphs and claw-free graphs. They have shown that finding minimum $b$-disjunctive dominating set problem is NP-complete for planar and bipartite graphs and also designed a dynamic programming based linear time algorithm to find a minimum b-disjunctive dominating set in a tree. In \cite{Henning2014}, Henning et al. have studied the relation between domination number and disjunctive domination number of a tree $T$ and proved that $\gamma(T)\leq 2 \gamma_2^d(T)-1$. They have also given a constructive characterization of the trees achieving equality in this bound. On the other hand, a variation of disjunctive domination
is also studied in the literature (see \cite{total_disj}).

In this paper, our focus is on algorithmic study of disjunctive domination problem. The rest of the paper is organized as follows. In Section \ref{sec:preliminaries}, we give some pertinent definitions and notations that would be used in the rest of the paper. In this section, we also  observe some graph classes where domination problem is NP-complete but disjunctive domination can be easily solved and vice versa. This motivates us to study the status of the problem in other graph classes. In Section \ref{sec:polyalgo}, we design a linear time algorithm for disjunctive domination problem in proper interval graphs, an important subclass of chordal graphs. In Section \ref{sec:NPC}, we prove that DDDP remains NP-complete for chordal graphs. In Section \ref{sec:approx}, we design a polynomial time approximation algorithm for MDDP for general graph $G$ with approximation ratio $\ln(\Delta^{2}+\Delta+2)+1$, where $\Delta$ is the maximum degree of $G$. In this section, we also prove that MDDP can not be approximated within $(1-\epsilon) \ln(|V|)$ for any $\epsilon>0$ unless NP $\subseteq$ DTIME$(|V|^{O(\log \log |V|)})$.  In addition, for bipartite graphs with maximum degree $3$, MDDP is shown to be APX-complete in this section. Finally, Section \ref{sec:conclu} concludes the paper.

\section{Preliminaries}
\label{sec:preliminaries}

\subsection{Notations}
Let $G =(V, E)$ be a  graph. Let $N^{2}_{G}(v)$ denote the set of vertices which
 are at distance $2$ from the vertex $v$ in graph $G$. Let $G[S]$, $S \subseteq V$ denote the
induced subgraph of $G$ on the vertex set $S$.  The degree of a
vertex $v \in V(G)$, denoted by $d_G(v)$, is the number of neighbors
of $v$, that is, $d_G(v)=|N_G(v)|$. The \emph{minimum degree} and
\emph {maximum degree} of a graph $G$ is defined by
$\delta(G)=\min_{v\in V(G)} d_G(v)$ and $\Delta(G)=\max_{v\in V(G)}
d_G(v)$, respectively. A set $S\subseteq V$ is called an independent
set of a graph $G=(V,E)$ if $uv\notin E$ for all $u,v\in S$. A set
$K\subseteq V$ is called a clique of a graph $G=(V,E)$ if $uv\in E$
for all $u,v\in K$. A set $C\subseteq V$ is called a \emph{vertex
cover} of a graph $G=(V,E)$ if for each edge $ab\in E$, either $a\in
C$ or $b\in C$. Let $n$ and $m$ denote the
number of vertices and number of edges of $G$, respectively. In this paper,
we only consider connected graphs with at least two vertices.

\subsection{Graph Classes}
A graph $G$ is said to be a \emph{chordal graph} if every cycle in
$G$ of length at least four has a chord, that is, an edge joining two
non-consecutive vertices of the cycle.  Let $\mathscr{F}$ be a
family of sets. The intersection graph of $\mathscr{F}$ is obtained
by taking each set in $\mathscr{F}$ as a vertex and joining two sets
in $\mathscr{F}$ if and only if they have a non-empty intersection.
A graph $G$ is an \emph{interval graph} if $G$ is the intersection
graph of a family $\mathscr{F}$ of intervals on the real line.  A
graph  $G$ is called a \emph{proper interval} graph if it is the
intersection graph of a family $\mathscr{F}$ of intervals on the
real line such that no interval in $F$ contains another interval in
$\mathscr{F}$ set theoretically.  A vertex $v\in V(G)$ is a
\emph{simplicial} vertex of $G$ if $N_G[v]$ is a clique of $G$. An
ordering $\alpha=(v_1,v_2,...,v_n)$ is a {\it perfect elimination
ordering} (PEO) of $G$ if $v_i$ is a simplicial vertex of
$G_i=G[\{v_i,v_{i+1},...,v_n\}]$ for all $i$, $1\leq i\leq n$. A graph $G$ has a PEO if and only if $G$ is chordal~\cite{fulkerson}. A PEO
$\alpha=(v_1,v_2,\ldots,v_n)$ of a chordal graph  is a
\emph{bi-compatible elimination ordering} (BCO) if
$\alpha^{-1}=(v_n,v_{n-1},\ldots,v_1)$, i.e., the reverse of
$\alpha$, is also a PEO of $G$. A graph $G$ has a BCO if and only if $G$ is a proper interval graph~\cite{jamison}. A graph $G=(V,E)$ is called a \emph{split graph} if its vertex set, $V$, can be partitioned into two sets, say
$X$ and $Y$, such that $X$ is an independent set and $Y$ is a clique
of $G$.

\subsection{Domination vs disjunctive domination}

In this subsection, we make some observations on complexity difference of domination and disjunctive domination problem. It is known that domination problem is NP-complete for split graphs \cite{Bertossi} and for graphs with diameter two \cite{ambalath}. But disjunctive domination problem can be easily solved in these graph classes. Because, disjunctive domination number is at most $2$ in these classes and $\gamma_2^d(G)=1$ if and only if $G$ contains a vertex of degree $n-1$. Next, we define a graph class, called \emph{GC graph}, for which domination problem is easily solvable, but disjunctive domination problem is NP-complete.
\begin{definition}[GC graph]
A graph $G'=(V',E')$ is said to be a \emph{GC graph} if it can be constructed from a general graph $G=(V,E)$ by adding a pendant vertex to every vertex of $G$. Formally, $V'=V\cup \{w_i \mid 1\leq i \leq n\}$ and $E'=E\cup \{v_{i}w_{i}\mid 1\leq i \leq n\}$.
\end{definition}

Note that, every vertex of a GC graph $G'$ is either a pendant vertex or adjacent to a unique pendant vertex and hence, $\gamma(G')=n$. In Section \ref{sec:NPC}, we show that DDDP is NP-complete for the class of GC graphs.

\section{Polynomial time algorithm for proper interval graphs}\label{sec:polyalgo}
In this section, we present a polynomial time algorithm to find a minimum cardinality disjunctive dominating set in proper interval graphs.

Let $\alpha=(v_{1},v_{2},\ldots,v_{n})$ be a BCO of the proper interval graph $G$. Let $MaxN_{G}(v_{i})$ denote the maximum index neighbor of $v_{i}$ with respect to the ordering $\alpha$. We start with an empty set $D$. At each iteration $i$ of the algorithm, we update the set $D$ in such a way that the vertex $v_{i}$ and all the vertices which appear before $v_{i}$ in the BCO $\alpha$, are disjunctively dominated by the set $D$. At the end of $n^{th}$ iteration, $D$ disjunctively dominate all the vertices of graph $G$.
The algorithm DISJUNCTIVE-PIG for finding a minimum cardinality disjunctive dominating set in a proper interval graph is given below.

\begin{algorithm}
%\DontPrintSemicolon
\caption{DISJUNCTIVE-PIG($G, \alpha=(v_{1},v_{2},\ldots,v_{n})$)}
%\LinesNumbered
%\begin{algorithmic}[1]
 Initialize $D=\emptyset$;\\
\For {$i=1:n$} {
Compute $N_{G}(v_{i})\cap D$ and $N_{G}^{2}(v_{i})\cap D$;\\
 \textbf{Case $1$:} Either $N_{G}[v_{i}]\cap D \neq \emptyset$, or $|N^{2}_{G}(v_{i})\cap D|\geq 2$\\
 \hspace*{1cm} No update in $D$ is done;\\
 \textbf{Case $2$:} $N_{G}[v_{i}]\cap D == \emptyset$ and $N^{2}_{G}(v_{i})\cap D==\emptyset$\\
 \hspace*{1cm} Update $D$ as $D=D\cup \{MaxN_{G}(v_{i})\}$;\\
 \textbf{Case $3$:} $N_{G}[v_{i}]\cap D == \emptyset$ and $|N^{2}_{G}(v_{i})\cap D|==1$\\
 \hspace*{1cm} Find $v_{r}\in N^{2}_{G}(v_{i})\cap D$;\\
 \hspace*{1cm} $v_{j}=Max[v_{i}]$; $v_{k}=Max[v_{j}]$;\\
 \hspace*{1cm} $S=\{v_{i+1},v_{i+2},\ldots,v_{j-1}\}$;\\
 \hspace*{1cm} \textbf{Subcase $3.1$:} For every $v\in S$, either $vv_{k}\in E$ or $d(v,v_{r})=2$\\
 \hspace*{2.4cm} Update $D$ as $D=D\cup \{v_{k}\}$;\\
 \hspace*{1cm} \textbf{Subcase $3.2$:} $v_{s}$ is the least index vertex in $S$ such that \\
\hspace*{1cm} $d(v_{s},v_{k})=2$ and $d(v_{s},v_{r})>2$\\
 \hspace*{2.4cm} Update $D$ as $D=D\cup \{MaxN_{G}(v_{s})\}$;
}
return $D$;
\end{algorithm}

Next we give the proof of correctness of the algorithm. Let $\alpha=(v_{1},v_{2},\ldots,v_{n})$ be the BCO of a proper interval graph $G$. Define the set $V_{i}=\{v_{1},v_{2},\ldots,v_{i}\}$, $1\leq i \leq n$, and $V_{0}=\emptyset$. Also suppose that $D_{i}$ denotes the set $D$ obtained after processing vertex $v_{i}$, $1\leq i \leq n$, and $D_{0}=\emptyset$. We will prove that $D_{n}$ is a minimum cardinality disjunctive dominating set of $G$.

\begin{theorem}
For each $i$, $0\leq i \leq n$, the following statements are true:
\begin{itemize}
\item[$(a)$] $D_{i}$ disjunctively dominates the set $V_{i}$.
\item[$(b)$] There exists a minimum cardinality disjunctive dominating set $D_{d}^{*}$ such that $D_{i}$ is contained in $D_{d}^{*}$.
\end{itemize}
\end{theorem}
\begin{proof}
We prove the theorem by induction on $i$. The basis step is trivial
as $D_{0}=\emptyset$. Next assume that the theorem is true for
$i-1$. So,  $(a)$ $D_{i-1}$ disjunctively dominates the set
$V_{i-1}$,  $(b)$ there exists a
minimum cardinality disjunctive dominating set $D_{d}^{*}$ such that $D_{i-1}$ is contained
in $D_{d}^{*}$.

Next we prove the theorem for $i$. According to our algorithm, we
need to  discuss the following three cases.

\noindent \textbf{Case 1:} Either $N_{G}[v_{i}]\cap D_{i-1} \neq \emptyset$, or $|N^{2}_{G}(v_{i})\cap D_{i-1}|\geq 2$.

Here $D_{i}=D_{i-1}$. It is easy to notice that all the conditions of the theorem are satisfied.

\noindent \textbf{Case 2:} $N_{G}[v_{i}]\cap D_{i-1} = \emptyset$ and $N^{2}_{G}(v_{i})\cap D_{i-1}=\emptyset$.

Here $D_{i}=D_{i-1}\cup \{v_{j}\}$ where $v_{j}=MaxN_{G}(v_{i})$. Hence, condition $(a)$ of
the theorem is trivially satisfied. If $v_{j}\in D_{d}^{*}$, then $D_{i}\subseteq D_{d}^{*}$.
Hence both the conditions of the theorem are satisfied, and $D_{d}^{*}$ is the required minimum
cardinality disjunctive dominating set of $G$. If $v_{j}\notin D_{d}^{*}$, then there are two possibilities:\\
\textbf{(I) There exists a vertex $v_{p}\in N_{G}[v_{i}]\cap D_{d}^{*}$. }

\noindent Define the set $D_{d}^{**}=(D_{d}^{*}\setminus \{v_{p}\})\cup \{v_{j}\}$. Note that $D_{i}\subseteq D_{d}^{**}$, and $|D_{d}^{*}|=|D_{d}^{**}|$. Now, to prove condition $(b)$ of the theorem, it is enough to show that $D_{d}^{**}$ is a disjunctive dominating set of $G$. Note that $D_{i-1}\cup \{v_{j}\}\subseteq D_{d}^{**}$. Now consider an arbitrary vertex $v_{a}$ of $G$. If $a<i$, then the vertex $v_{a}$ is disjunctively dominated by the set $D_{i-1}$, and hence by $D_{d}^{**}$. If $a\geq i$, and $v_{p}\in N_{G}[a]$, then $v_{j}\in N_{G}[v_{a}]$. If $a\geq i$, and $v_{p}\in N_{G}^{2}(v_{a})$, then $v_{j}\in N_{G}[v_{a}]$ or $v_{j}\in N_{G}^{2}(v_{a})$. This proves that $D_{d}^{**}$ is a disjunctive dominating set of $G$.

\noindent \textbf{(II) For $q<s$, vertices $v_{q},v_{s}\in N^{2}_{G}(v_{i})\cap D_{d}^{*}$.} \\
Let $MaxN_{G}(v_{i})=v_{j}$ and $MaxN_{G}(v_{j})=v_{k}$. Then $q<s\leq k$. Let $v_{t}=MaxN_{G}(v_{s})$ and $v_{r}=MaxN_{G}(v_{t})$. We again consider three possibilities:\\
$(i)$ $q<s<i$\\
Here $r\leq j$. Now consider an arbitrary vertex $v_{a}$ of $G$. If $a<i$, then the vertex $v_{a}$ is disjunctively dominated by the set $D_{i-1}$. If $a\geq i$, and $v_{s}\in N^{2}_{G}(v_{a})$ or $v_{q},v_{s}\in N^{2}_{G}(v_{i})$, then $v_{j}\in N_{G}[v_{a}]$. Hence $(D_{d}^{*}\setminus\{v_{q},v_{s}\})\cup \{v_{j}\}$ is a disjunctive dominating set of $G$ of cardinality less than $|D_{d}^{*}|$, which is a contradiction, as $D_{d}^{*}$ is a minimum disjunctive dominating set of $G$. Therefore, this situation will never arise.\\
$(ii)$ $q<i<s$\\
Consider an arbitrary vertex $v_{a}$ of $G$. If $a<i$, then the vertex $v_{a}$ is disjunctively dominated by the set $D_{i-1}$. If $a\geq i$, and $v_{q}\in N_{G}^{2}(v_{a})$, then $v_{j}\in N_{G}[v_{a}]$. If $a\geq i$, and $v_{q}\notin N_{G}^{2}(v_{a})$, and either $v_{s}\in N_{G}[v_{a}]$ or $v_{s}\in N_{G}^{2}(v_{a})$, then either $v_{j}\in N_{G}[v_{a}]$ or $v_{t}\in N_{G}[v_{a}]$. Hence, if we define $D_{d}^{**}=(D_{d}^{*}\setminus \{v_{q},v_{s}\})\cup \{v_{j},v_{t}\}$, then $D_{d}^{**}$ is a minimum cardinality disjunctive dominating set of $G$ and $D_{i}\subseteq D_{d}^{**}$. This proves the condition $(b)$ of the theorem.\\
 $(iii)$ $i<q<s$\\
Here $s\leq k$. Consider an arbitrary vertex $v_{a}$ of $G$. If $a<i$, then the vertex $v_{a}$ is disjunctively dominated by the set $D_{i-1}$. If $a\geq i$, and $v_{q}\in N_{G}[v_{a}]$ or $v_{s}\in N_{G}[v_{a}]$ or $v_{q},v_{s}\in N_{G}^{2}(v_{a})$ or $v_{s}\in N_{G}^{2}(v_{a})$, then either $v_{j}\in N_{G}[v_{a}]$ or $v_{t}\in N_{G}[v_{a}]$. Hence, if we define $D_{d}^{**}=(D_{d}^{*}\setminus \{v_{q},v_{s}\})\cup \{v_{j},v_{t}\}$, then $D_{d}^{**}$ is a minimum cardinality disjunctive dominating set of $G$ and $D_{i}\subseteq D_{d}^{**}$. This proves the condition $(b)$ of the theorem.

\noindent \textbf{Case 3:} $|N^{2}_{G}(v_{i})\cap D_{i-1}|=1$, $N^{2}_{G}(v_{i})\cap D_{i-1}=\{v_{r}\}$ $(r<i)$, $v_{j}=MaxN_{G}(v_{i})$, $v_{k}=MaxN_{G}(v_{j})$, and $S=\{v_{i+1},v_{i+2},\ldots,v_{j-1}\}$.

\noindent \textbf{Subcase 3.1:} For every $v\in S$, either $vv_{k}\in E$ or $d(v,v_{r})=2$.\\
 Here $D_{i}=D_{i-1}\cup \{v_{k}\}$.

Clearly, condition $(a)$ of the theorem is satisfied. If $v_{k}\in D_{d}^{*}$, then $D_{i}\subseteq D_{d}^{*}$.
Hence both the conditions of the theorem are satisfied, and $D_{d}^{*}$ is the required minimum
cardinality disjunctive dominating set of $G$. If $v_{k}\notin D_{d}^{*}$, then to disjunctively dominate $v_{i}$, at least one vertex before $v_{k}$  in BCO $\alpha$, say $v_{p}$, must belong to $D_{d}^{*}$. Define $D_{d}^{**}=(D_{d}^{*}\setminus \{v_{p}\})\cup \{v_{k}\}$. Then $|D_{d}^{**}|=|D_{d}^{*}|$ and $D_{i}\subseteq D_{d}^{**}$. Now, to prove condition $(b)$ of the theorem, it is enough to show that $D_{d}^{**}$ is a disjunctive dominating set of $G$. Consider an arbitrary vertex $v_{b}$ in $G$. If $b\leq k$, then $v_{b}$ is disjunctively dominated by the set $D_{i-1}\cup \{v_{k}\}$, and hence by $D_{d}^{**}$. If $b>k$, and $v_{p}\in N_{G}[v_{b}]$, then $v_{k}\in N_{G}[v_{b}]$. If $b>k$, and $v_{p}\in N_{G}^{2}(v_{b})$, then either $v_{k}\in N_{G}[v_{b}]$ or $v_{k}\in N_{G}^{2}(v_{b})$. Hence $D_{d}^{**}$ is a disjunctive dominating set of $G$.

\noindent \textbf{Subcase 3.2:}  $v_s$ is the least index vertex in $S$ such that $d_{G}(v_s, v_k)=2$ and
$d(v_s,v_r) > 2$.\\ Here $D_{i}= D_{i-1}\cup \{v_{p}\}$, where $v_{p}=MaxN_{G}(v_s)$. Clearly, condition $(a)$ of
the theorem is trivially satisfied. If $v_{p}\in D_{d}^{*}$, then $D_{i}\subseteq D_{d}^{*}$.
Hence both the conditions of the theorem are satisfied, and $D_{d}^{*}$ is the required minimum
cardinality disjunctive dominating set of $G$. If $v_{p}\notin D_{d}^{*}$, then there are two possibilities:\\
\textbf{(I) $v_{q}\in D_{d}^{*}\setminus D_{i-1}$, where $q<p$}\\
Define $D_{d}^{**}=(D_{d}^{*}\setminus \{v_{q}\})\cup \{v_{p}\}$. Then $|D_{d}^{**}|=|D_{d}^{*}|$ and $D_{i}\subseteq D_{d}^{**}$. Now, to prove condition $(b)$ of the theorem, it is enough to show that $D_{d}^{**}$ is a disjunctive dominating set of $G$. Consider an arbitrary vertex $v_{b}$ in $G$. If $b< i$, then $v_{b}$ is disjunctively dominated by the set $D_{i-1}$, and hence by $D_{d}^{**}$. If $b\geq i$, and $v_{q}\in N_{G}[v_{b}]$, then  $v_{p}\in N_{G}[v_{b}]$ or $v_{r},v_{p}\in N_{G}^{2}(v_{b})$. If $b\geq i$, and $v_{q}\in N_{G}^{2}(v_{b})$, then $v_{p}\in N_{G}[v_{b}]$ or $v_{p}\in N_{G}^{2}(v_{b})$. Hence $D_{d}^{**}$ is a disjunctive dominating set of $G$.

\noindent \textbf{(II) $D_{d}^{*}\cap V_{p}=D_{i-1}$}\\
 To disjunctively dominate the vertex $v_{s}$, at least two vertices from the set $V_{w}\setminus (V_{p}\cup D_{i-1})$ must belong to $D_{d}^{*}$, where $v_{w}=MaxN_{G}(v_{p})=MaxN_{G}(MaxN_{G}(v_{s}))$. Let they are $v_{t1},v_{t2}$ where $t_{1}<t_{2}$. Note that $p<t_{1}<t_{2}\leq w$. Let $v_{w'}=MaxN_{G}(MaxN_{G}(v_{w}))$. Now define the set $D_{d}^{**}=(D_{d}^{*}\setminus \{v_{t1},v_{t2}\})\cup \{v_{p},v_{w'}\}$. Then $|D_{d}^{**}|=|D_{d}^{*}|$ and $D_{i}\subseteq D_{d}^{**}$. Now to prove the condition $(b)$ of the theorem, it is enough to show that $D_{d}^{**}$ a disjunctive dominating set of $G$. Consider  an arbitrary vertex $v_{b}$ in $G$. If $b< i$, then $v_{b}$ is disjunctively dominated by the set $D_{i-1}$, and hence by $D_{d}^{**}$. If $s>b\geq i$, then either $v_{p}\in N_{G}[v_{b}]$ or $v_{r},v_{p}\in N_{G}^{2}(v_{b})$ (since every vertex in $V_{s-1}\setminus V_{i-1}$ is at distance $2$ from the vertex $v_{r}$). If $b\geq s$, and $v_{t_{1}}\in N_{G}[v_{b}]$ or $v_{t_{2}}\in N_{G}[v_{b}]$ or $v_{t_{1}},v_{t_{2}}\in N_{G}^{2}(v_{b})$, the either $v_{p}\in N_{G}[v_{b}]$ or $v_{p},v_{w'}\in N_{G}^{2}(v_{b})$.  If $b\geq s$ and $v_{t_{1}}\in N_{G}^{2}(v_{b})$ and $v_{t_{2}}\notin N_{G}^{2}(v_{b})$, then $v_{p}\in N_{G}[v_{b}]$. If $b\geq s$ and $v_{t_{1}}\notin N_{G}^{2}(v_{b})$ and $v_{t_{2}}\in N_{G}^{2}(v_{b})$, then either $v_{w'}\in N_{G}[v_{b}]$ or $v_{p},v_{w'}\in N_{G}^{2}(v_{b})$. Hence $D_{d}^{**}$ is a disjunctive dominating set of $G$.
Hence our theorem is proved.
\end{proof}

In view of the above theorem, the set $D$ computed by the algorithm DISJUNCTIVE-PIG is a minimum cardinality
 disjunctive dominating set of $G$.  Now, we show that the algorithm DISJUNCTIVE-PIG can be implemented in polynomial time. We use the adjacency list representation of the graph. We maintain an array $D_{set}$ for the set $D$ such that $D_{set}[j]=1$ if $v_{j}\in D$. We maintain the All pair distance Matrix $Dist[1..n,1..n]$ such that $Dist[i,j]$ is the distance between $v_i$ and $v_j$. This can be done in $O(n^3)$ time. Now $N_G[v_i] \cap D$ can be computed in $O(n)$ time by looking up $Dist$ matrix and array $D_{set}$. Similarly, $N^{2}_G(v_i) \cap D$ can be computed in $O(n)$ time. Also $MaxN_{G}(v_{i})$ can be computed in $O(n)$ time. Hence, in any iteration,  all the operations can be done in $O(n^2)$ time. Therefore overall time is $O(n^3)$, as number of iterations are $n$.
 Since, BCO of a proper interval graph can be computed in $O(n+m)$ time \cite{panda}, and all the computations in the algorithm DISJUNCTIVE-PIG can be done in $O(n^{3})$ time, we have the following theorem.

\begin{theorem}
MDDP can be solved in $O(n^{3})$ time in proper interval graphs.
\end{theorem}

However, the algorithm DISJUNCTIVE-PIG can be implemented in $O(n+m)$ time using additional data structures. The details are given below.
 We first describe some notations. Let $\alpha=(v_{1},v_{2},\ldots,v_{n})$ be a BCO of the proper interval graph $G=(V,E)$. We maintain a set $D$. Initially $D=\emptyset$. At the end of $n^{th}$ iteration, $D$ becomes a minimum cardinality disjunctive dominating set of $G$.
 We maintain two arrays $Min[1,\ldots,n]$ and $Max[1,\ldots,n]$. For a vertex $v$, $Min[v]$ denotes the minimum index vertex in the BCO $\alpha$, which is adjacent to $v$, and $Max[v]$ denotes the maximum index vertex in the BCO $\alpha$, which is adjacent to $v$.
 We also maintain an array $D_{count}[1,\ldots,n]$. For a vertex $v\in V$, $D_{count}[v]$ denotes the number of vertices in $D$ which dominate the vertex $v$.

\begin{lemma} The following statements are true:
  \begin{itemize}
  \item[(i)] $D_{count}[v_{i}]=|N_{G}[v_{i}]\cap D|$.
  \item[(ii)] If $N_{G}[v_{i}]\cap D=\emptyset$, then $D_{count}[Max[v_{i}]]+D_{count}[Min[v_{i}]]=|N_{G}^{2}(v_{i})\cap D|$.
 \end{itemize}
\end{lemma}
\begin{proof}
The proof is easy and hence is omitted.
\end{proof}

Based on the above discussion, we have the detailed algorithm for finding minimum cardinality disjunctive dominating set which is presented in M-DISJUNCTIVE-PIG.

\begin{algorithm}[H]
%\DontPrintSemicolon
\caption{M-DISJUNCTIVE-PIG(G)}
\LinesNumbered
 Obtain a BCO $\sigma=\{v_{1},v_{2},...,v_{n}\}$ of proper interval graph $G;$\\
Obtain array $Min$ and $Max$;\\
Initialize $D=\emptyset$;\\
Initialize  $D_{count}[v_{i}]=0$ for all $i$, $1\leq i \leq n$;\\
\For {$i=1:n$} {
\If {($(D_{count}[v_{i}]!=0)$ or $(D_{count}[Max[v_{i}]]+D_{count}[Min[v_{i}]]\geq 2))$}{
no update;
}
\ElseIf{$(D_{count}[v_{i}]==0)$ and $(D_{count}[Max[v_{i}]]+D_{count}[Min[v_{i}]]==0)$}
{
$v_{k}=Max[v_{i}]$;\\
$D=D\cup \{v_{k}\}$;\\
\ForEach{$v\in N_{G}[v_{k}]$}{
$D_{count}[v]=D_{count}[v]+1$;
}
}
\ElseIf{($(D_{count}[v_{i}]==0)$ and $(D_{count}[Max[v_{i}]]+D_{count}[Min[v_{i}]]==1)$)}
{
(This basically means that $D_{count}[Min[v_{i}]]=1$)\\
Let $v_{t}=Min[v_{i}]$,
$v_{j}=Max[v_{i}]$, and
$v_{k}=Max[v_{j}]$;\\
Let $\{v_{r}\}=N_{G}[v_{t}]\cap D$;\\
%\ForEach{$v\in N_{G}[v_{t}]$}
%{
%\If{$v\in D$}
%{
%$v_{r}=v$;
%}
%}
\For {$s=i+1:j-1$} {
Let $v_{a}=Min[v_{k}]$, $v_{b}=Max[v_{r}]$, $v_{c}=Min[v_{s}]$;\\
\If{$s<a$ and $b<c$}{
$D=D\cup \{Max[v_{s}]\}$;\\
\ForEach{$v\in N_{G}[Max[v_{s}]]$}{
$D_{count}[v]=D_{count}[v]+1$;
}
break;

}
}
\If{$s==j$}{
$D=D\cup \{v_{k}\}$;\\
\ForEach{$v\in N_{G}[v_{k}]$}{
$D_{count}[v]=D_{count}[v]+1$;
}
}
}
}
return $D$;
\end{algorithm}

 Next we show that this algorithm M-DISJUNCTIVE-PIG can be implemented in $O(n+m)$ time. We first compute $Max[v_i]$ and $Min[v_i]$ for each $v_i$, $1\leq i \leq n$. This takes $O(d_{G}(v_i))$ time for each vertex $v_i$. Hence arrays $Min$ and $Max$ can be computed in $O(n+m)$ time. We can find a vertex in $N_G[v_t] \cap D$ in $O(1)$ time by maintaining  an array $B[1,\ldots,n]$ of linked lists such that $B[i]$ contains all the vertices of $N_G[v_i] \cap D$. This is done by inserting $v_j$ in the linked lists of $v_{j}$ and all the neighbors of $v_j$ whenever $v_j$ is included in $D$. So maintaining this information takes $\sum_{v\in D}(d(v))= O(n+m)$ time. Therefore, all the computations in the algorithm M-DISJUNCTIVE-PIG can be done in $\sum_{i=1}^{n}(d_{G}(v_i)) + \sum_{v\in D}(d_{G}(v))= O(n+m)$ time.

In view of this, we have the following theorem.

\begin{theorem}
The algorithm M-DISJUNCTIVE-PIG can be implemented in $O(n+m)$ time and hence MDDP can be solved in $O(n+m)$ time in proper interval graphs.
\end{theorem}

\section{NP-completeness}\label{sec:NPC}
In this section, we prove that DDDP is NP-complete for chordal graphs. For that, we first show that DDDP is NP-complete for GC graphs. To prove this NP-completeness result, we use a reduction from another variant of domination problem, namely \emph{$2$-domination problem}. For a graph $G=(V,E)$, a set $D_2\subseteq V$ is called \emph{$2$-dominating set} if every vertex $v\in V\setminus D_2$ has at least two neighbors in $D_2$. Given a positive integer $k$ and a graph $G=(V,E)$, the \textsc{$2$-domination Decision Problem} ($2$DDP) is to decide whether $G$ has a $2$-dominating set of cardinality at most $k$. It is known that 2DDP is NP-complete for chordal graphs \cite{jacobson}. The following lemma shows that DDDP is NP-complete for GC graphs.

\begin{lemma}\label{lem:DDDP_GC}
DDDP is NP-complete for GC graphs.
\end{lemma}
\begin{proof}
Clearly, DDDP is in NP for GC graphs. To prove the NP-hardness, we give a polynomial transformation from 2DDP for general graphs. Let $G=(V,E)$ and $k$ be an instance of 2DDP. Given a graph $G=(V,E)$ where $V=\{v_{1},v_{2},\ldots,v_{n}\}$, we construct the graph $G'=(V',E')$ in the following way: $V'=V\cup \{w_{i}\mid 1\leq i \leq n\}$ and $E'=E\cup \{v_{i}w_{i}\mid 1\leq i \leq n\}$. Clearly $G'$ is a GC graph and it can be constructed from $G$ in polynomial time.

The following claim is enough to complete the proof of the theorem.

%\noindent\textbf{Claim 1}
\begin{claim}
 $G$ has a $2$-dominating set of cardinality at most $k$ if and only if $G'$ has a disjunctive dominating set of cardinality at most $k$.
\end{claim}
\begin{proof}
(Proof of the claim)
Let $D_2$ be a $2$-dominating set of $G$ of cardinality at most $k$. Clearly $D_2$ is a disjunctive dominating set of $G'$. Because every $v_i\in V'$ either is in $D_2$ or dominated by at least two vertices of $D_2$ and every $w_i\in V'$ is either dominated by $v_i\in D_2$ or contains at least two vertices from $D_2$ at a distance of two. Hence, $G'$ has a disjunctive dominating set of cardinality at most $k$.

Conversely, suppose that $D_d$ is a disjunctive dominating set of $G'$ of cardinality at most $k$. Note that, every vertex of $G'$ is either a pendant vertex or a support vertex. Also, the vertex set of graph $G$ is exactly the set of all support vertices of $G'$. Let $P$ be the set of pendant vertices of graph $G'$, i.e., $P=\{w_i\mid 1\leq i\leq n\}$. If a pendant vertex $w_{i}\in D_{d}$, then the set $D'_{d}=(D_{d}\setminus\{w_{i}\})\cup \{v_{i}\}$ still remains a disjunctive dominating set of $G'$ of cardinality at most $k$. So, without loss of generality we assume that $D_{d}\cap P=\emptyset$. Now for every vertex $v_{i}\in V$, either $v_{i}\in D_{d}$ or $|N_{G}(v_{i})\cap D_{d}|\geq 2$. If not, let there is a vertex $v_i\in V\setminus D_d$ such that $|N_{G}(v_{i})\cap D_{d}|\leq 1$. This implies that the vertex $w_i\in V'$ is neither dominated nor has at least two vertices from $D_d$ at a distance of two, contradicting the fact that $D_d$ is a disjunctive dominating set of $G'$. Hence, $D_d$ is a $2$-dominating set of $G$ of cardinality at most $k$.
\end{proof}

Hence, it is proved that DDDP is NP-complete for GC graphs.
\end{proof}

It is easy to observe that, if the graph $G$ is chordal, then the constructed graph $G'$ in Lemma \ref{lem:DDDP_GC} is also chordal. Hence, we have the main result of this section as a corollary.
\begin{theorem}
DDDP is NP-complete for chordal graphs.
\end{theorem}

\section{Approximation results}\label{sec:approx}

\subsection{Approximation algorithm}\label{ssec:approxalgo}
In this subsection, we propose a  $(\ln(\Delta^{2}+\Delta+2)+1)$-approximation algorithm for MDDP. Our algorithm is based on the reduction from MDDP to the \textsc{Constrained Multiset Multicover} (CMSMC) problem. We first recall the definition of the \textsc{Constrained Multiset Multicover} problem.

Let $X$ be a set and $\mathcal{F}$ be a collection of subsets of $X$. The \textsc{Set Cover} problem is to find a smallest sub-collection, say $\mathcal{C}$ of $\mathcal{F}$, such that $\mathcal{C}$ covers all the elements of $X$, that is, $\cup_{S\in \mathcal{C}}S=X$. The \textsc{Constrained Multiset Multicover} problem is a generalization of the \textsc{Set Cover} problem. In this problem, $\mathcal{F}$ is the collection of multisets of $X$, that is, each element $x\in X$ occurs in a multiset $S\in \mathcal{F}$ with arbitrary multiplicity, and each element $x\in X$ has an integer coverage requirement $r_{x}$ which specifies how many times $x$ has to be covered. Note that each set $S\in \mathcal{F}$ is chosen at most once. So, for a given set $X$, a collection $\mathcal{F}$ of multisets of $X$, and integer requirement $r_{x}$ for each $x\in X$, the CMSMC problem is to find a smallest collection $\mathcal{C}\subseteq \mathcal{F}$, such that $\mathcal{C}$ covers each element $x$ in $X$ at least $r_{x}$ times. In the case, when $r_{x}$ is constant for each $x\in X$, then $\mathcal{C}$ is called a $r_{x}$-cover of $X$, and the CMSMC problem is to find a minimum cardinality $r_{x}$-cover of $X$.
\begin{theorem}\label{th53}
The \textsc{Minimum Disjunctive Domination Problem} for a graph $G=(V,E)$ with maximum degree $\Delta$ can be approximated with an approximation ratio of $\ln(\Delta^{2}+\Delta+2)+1$.
\end{theorem}

\begin{proof}
Let us show the transformation from MDDP to the CMSMC problem.

\noindent \textbf{Construction :} Let $G=(V,E)$ be a graph with $n$ vertices and $m$ edges where $V=\{v_{1},v_{2},\ldots,v_{n}\}$ (an instance of MDDP). Now we construct an instance of the CMSMC problem, that is, a set $X$, a family $\mathcal{F}$ of multisets of $X$, and a vector $R=(r_{x})_{x\in X}$ ($r_{x}$ is a non-negative integer for each $x\in X$) in the following way:\\
% ordered pair $(X,\mathcal{F})$, an instance of the MS2C problem in the following way:
  $X=V$,
   $\mathcal{F}=\{F_{1},F_{2},\ldots,F_{n}\}$, where for each $i$, $1\leq i \leq n$, $F_{i}$ is a multiset which contains two copies of each element in $N_{G}[v_{i}]$ and one copy of the set of elements which are at distance $2$ from the vertex $v_{i}$ in graph $G$,
$r_{x}=2$ for each $x\in X$.

Now we first prove the following correspondence.

\begin{claim} \label{th52}
The set  $D=\{v_{i_{1}},v_{i_{2}},\ldots,v_{i_{k}}\}$ is a disjunctive dominating set of $G$ if and only if $\mathcal{C}=\{F_{i_{1}},F_{i_{2}},\ldots,F_{i_{k}}\}$ is a 2-cover of $X$.
\end{claim}
\begin{proof} (Proof of the claim)
Suppose $D=\{v_{i_{1}},v_{i_{2}},\ldots,v_{i_{k}}\}$  is a disjunctive dominating set of $G$. Let $\mathcal{C}=\{F_{i_{1}},F_{i_{2}},\ldots,F_{i_{k}}\}$, We want to show that $\mathcal{C}$ is a 2-cover of $X$, that is, each element $v\in X$ is 2-covered by $\mathcal{C}$. Consider an arbitrary element $v\in X$. Note that $X=V$. If either $v$ or one of its neighbor belongs to $D$, that is, $v_{i_{r}}\in N_{G}[v]\cap D$, then the set $F_{i_{r}}$ contains 2 copies of $v$, and hence $v$ is 2-covered. If $N_{G}[v]\cap D=\emptyset$, then $|N_{G}^{2}(v)\cap D|\geq 2$. Let $v_{i_{p}},v_{i_{q}}\in N_{G}^{2}(v)\cap D$. Then each $F_{i_{p}}$ and $F_{i_{q}}$ contains a copy of $v$, and hence $v$ is 2-covered. Hence $\mathcal{C}$ is a $2$-cover of $X$.

Conversely, suppose that $\mathcal{C}=\{F_{i_{1}},F_{i_{2}},\ldots,F_{i_{k}}\}$ is a 2-cover of $X$. Let $D=\{v_{i_{1}},v_{i_{2}},\ldots,v_{i_{k}}\}$. We want to show that $D$ is a disjunctive dominating set of $G$. Consider any arbitrary vertex $v\in V$. Then $v\in X$ (as $X=V$). Hence $v$ is 2-covered by $\mathcal{C}$. Then, we have two possibilities: $(i)$ There exists a set $F_{i_{r}}\in \mathcal{C}$, which contains $2$ copies of $v$. In this case, $v_{i_{r}}$ is either $v$ or one of the neighbor of $v$, and hence $v$ is disjunctively dominated by the set $D$. $(ii)$ There exists two sets $F_{i_{p}},F_{i_{q}}\in C$, each containing a copy of $v$. Then $v_{i_{p}}$ and $v_{i_{q}}$ both are at distance $2$ from the vertex $v$. Hence again $v$ is disjunctively dominated by the set $D$. This proves that $D$ is a disjunctive dominating set of $G$.

This completes the proof of the claim.
\end{proof}

By the above claim, if $D_{d}^{*}$ is a minimum cardinality disjunctive dominating set of $G$ and $\mathcal{C^{*}}$ is an optimal $2$-cover of $X$, then $|D_{d}^{*}|=|\mathcal{C^{*}}|$.
In \cite{rajgopalan}, S. Rajgopalan and V. V. Vazirani gave a greedy approximation algorithm for the CMSMC problem, which achieves an approximation ratio of $\ln(|F_{M}|)+1$, where $F_{M}$ is the maximum cardinality multiset in $\mathcal{F}$. Let $\mathcal{C^{*}}$ be an optimal 2-cover and $\mathcal{C}'$ be a 2-cover obtained by greedy approximation algorithm, then $|\mathcal{C}'|\leq (\ln(|F_{M}|)+1)\cdot |\mathcal{C^{*}}|$. Given a 2-cover of $X$, we can also obtain a disjunctive dominating set of graph $G$ of same cardinality. Suppose that $D_{d}'$ is a disjunctive dominating set of $G$ obtained from $2$-cover $\mathcal{C}'$ of $X$. Then $|D_{d}'|\leq (\ln(|F_{M}|)+1)\cdot |D_{d}^{*}|$. If the maximum degree of the graph $G$ is $\Delta$, then the cardinality of a set in family $\mathcal{C}$ will be at most $2(\Delta+1)+\Delta(\Delta-1)$, which is equal to $\Delta^{2}+\Delta+2$. Hence $|D_{d}'|\leq (\ln(\Delta^{2}+\Delta+2)+1)\cdot |D_{d}^{*}|$. This completes the proof of the theorem.
\end{proof}

\subsection{Lower bound on approximation ratio}
To obtain the lower bound, we give an approximation preserving
reduction from the \textsc{Minimum Domination} problem. The following approximation hardness result for the
\textsc{Minimum Domination} problem is already known.

\begin{theorem}\label{th:1}
\cite{chlebik} For a graph $G=(V,E)$, the \textsc{Minimum Domination} problem can not be approximated within $(1-\epsilon) \ln|V|$ for any $\epsilon > 0$ unless NP $\subseteq $ DTIME $(|V|^{O(\log\log|V|)})$.
\end{theorem}

\begin{theorem}
For a graph $G=(V,E)$, MDDP can not be approximated within $(1-\epsilon) \ln|V|$ for any $\epsilon > 0$ unless NP $\subseteq$ DTIME$(|V|^{O(\log\log|V|)})$.
\end{theorem}
\begin{proof}
Let us describe the reduction from the \textsc{Minimum Domination} problem to MDDP. Let $G=(V,E)$, where $V=\{v_{1},v_{2},\ldots,v_{n}\}$ be an instance of the \textsc{Minimum Domination} problem. Now, we construct a graph $H=(V_{H},E_{H})$ an instance of MDDP in the following way:
$V_{H}=V\cup \{w_{i},z_{i}\mid 1\leq i \leq n\}\cup \{p,q\}$, $E_{H}=E\cup \{v_{i}w_{i},w_{i}z_{i},z_{i}p\mid 1\leq i \leq n\}\cup \{pq\}$.

Fig.~\ref{fig:2} illustrates the construction of the graph $H$ from a given graph $G$. Note that $|V_{H}|=3|V|+2$.
\begin{figure}[H]
  \begin{center}
  \includegraphics[width=6cm, height=3.5cm]{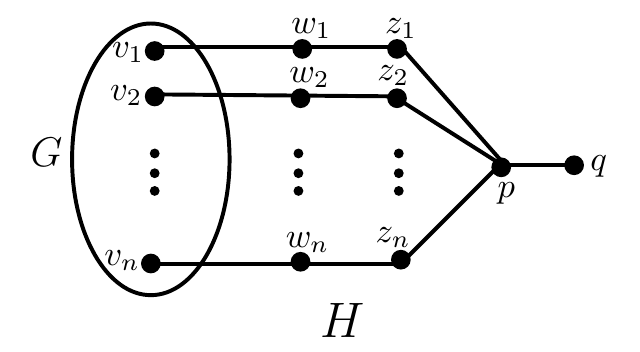}\\
 \caption{An illustration to the construction of $H$ from $G$}
\label{fig:2}
\end{center}
\end{figure}

If $D^{*}$ is a minimum cardinality dominating set of $G$, then $D^{*}\cup \{p\}$ is a disjunctive dominating set of $H$. Hence for a minimum cardinality disjunctive dominating set $D_{d}^{*}$ of $H$, $|D_{d}^{*}|\leq |D^{*}|+1$.

On the other hand, let $D_{d}$ be a disjunctive dominating set of $H$. Consider the vertex $w_{i}$. Since $w_{i}$ is disjunctively dominated by the set $D_{d}$, one of the following possibilities may occur:

\noindent $(i)$ $v_{i}\in D_{d}$, $(ii)$ $w_{i}\in D_{d}$ or $z_{i}\in D_{d}$, $(iii)$ $|N^{2}_{H}(w_{i})\cap D_{d}|\geq 2$, that is, $p\in D_{d}$ and $N_{G}(v_{i})\cap D_{d}\neq \emptyset$.

If $(ii)$ occurs, then define $D_{d}=(D_{d}\setminus \{w_{i},z_{i}\})\cup \{v_{i}\}$. Do it for all $i$, $1\leq i \leq n$. Note that the set $D=D_{d}\cap V$ dominates all the vertices of $G$, and $|D|\leq |D_{d}|$.

Now suppose that MDDP can be approximated with an approximation ratio of $\alpha$, where $\alpha=(1-\epsilon)\ln(|V_{H}|)$ for some fixed $\epsilon>0$, by a polynomial time approximation algorithm APPROX-DISJUNCTIVE. Let $l$ be a fixed positive integer. Consider the following algorithm to compute a dominating set of a given graph $G$.

\begin{algorithm}[H]
\caption{APPROX-DOMINATION(G)}
 \textbf{Input:} A graph $G=(V,E)$.\\
\textbf{Output:} A dominating set $D$ of graph $G$.\\
\Begin{
 \eIf {there exists a minimum dominating set $D'$ of cardinality $\leq l$}{
 $D=D'$;}
{
Construct the graph $H$;\\
Compute a disjunctive dominating set $D_{d}$ of $H$ using the algorithm \textsc{APPROX-DISJUNCTIVE};\\
\For{$i=1:m$}{
 \If{$w_{i}\in D_{d}$ or $z_{i}\in D_{d}$}
 {
 $D_{d}=(D_{d}\setminus \{w_{i},z_{i}\})\cup \{v_{i}\}$;
 }
 }
% each $i$, $1\leq i \leq m$, if $w_{i}\in D_{d}$ or $z_{i}\in D_{d}$, define $D_{d}=(D_{d}\setminus \{w_{i},z_{i}\})\cup \{v_{i}\}$; \\
$D=D_{d}\cap V$;
}
return $D$;
}
\end{algorithm}

Clearly, the algorithm APPROX-DOMINATION outputs a dominating set of $G$ in polynomial time. If the cardinality of a minimum dominating set of $G$ is at most $l$, then it can be computed in polynomial time. So, we consider the case, when the cardinality of a minimum dominating set of $G$ is greater than $l$. Let $D^{*}$ denotes a minimum cardinality dominating set of $G$, and $D_{d}^{*}$ denotes a minimum cardinality disjunctive dominating set of $H$. Note that $|D^{*}|>l$.

Let $D$ be the dominating set of $G$ computed by the algorithm APPROX-DOMINATION, then $|D|\leq |D_{d}|\leq \alpha |D_{d}^{*}|\leq \alpha (|D^{*}|+1)=\alpha(1+\frac{1}{|D^{*}|})|D^{*}|<\alpha(1+\frac{1}{l})|D^{*}|$.

Since $\epsilon$ is fixed, there exists a positive integer $l$ such that $\frac{1}{l}<\epsilon$. So, $|D|<\alpha (1+\epsilon)|D^{*}|= (1-\epsilon)(1+\epsilon)\ln(|V_{H}|)|D^{*}|=(1-\epsilon')\ln(|V_{H}|)|D^{*}|$. Since $|V_{H}|=3|V|+1$, and $|V|$ is very large, $\ln(|V_{H}|)\approx \ln (|V|)$. Hence $|D|<(1-\epsilon')\ln(|V|)|D^{*}|$. Hence, the dominating set $D$ computed by the algorithm APPROX-DOMINATION achieves an approximation ratio of $(1-\epsilon')\ln(|V|)$ for some $\epsilon'>0$.

By Theorem \ref{th:1}, if the \textsc{Minimum Domination} problem can be
approximated within a ratio of $(1-\epsilon')\ln(|V|)$, then NP
$\subseteq$ DTIME$(|V|^{O(\log\log |V|)})$. This proves that for a graph $H=(V_{H},E_{H})$,
MDDP can not be approximated within a ratio of $(1-\epsilon)\ln(|V_{H}|)$ unless NP $\subseteq$
DTIME$(|V_{H}|^{O(\log\log |V_{H}|)})$.
\end{proof}

\subsection{APX-completeness}\label{ssec:APX-complete}
In this subsection, we prove that MDDP is APX-complete for bounded degree graphs. To prove this, we need the concept of L-reduction, which is defined as follows.

\begin{definition}
Given two NP optimization problems $F$ and $G$ and a polynomial time
transformation $f$ from instances of $F$ to instances of $G$, we say
that $f$ is an L-reduction if there are positive constants $\alpha$
and $\beta$ such that for every instance $x$ of $F$
\begin{enumerate}
  \item $opt_{G}(f(x)) \leq  \alpha \cdot opt_{F}(x)$.
  \item for every feasible solution $y$ of $f(x)$ with objective value $m_{G}(f(x),y)=c_{2}$
we can in polynomial time find a solution $y'$ of $x$ with
$m_{F}(x,y')=c_{1}$ such that $|opt_{F}(x)-c_{1}| \leq \beta
|opt_{G}(f(x))-c_{2}|$.
\end{enumerate}
\end{definition}

To show the APX-completeness of a problem $\Pi \in $APX, it is enough to show that there is an L-reduction from some APX-complete problem to $\Pi$ \cite{Ausiello1999}.

By Theorem~\ref{th53}, it is clear that MDDP can be approximated within a constant factor for bounded degree graphs. Thus the problem is in APX for bounded degree graphs. To show the APX-hardness of MDDP, we give an L-reduction from the \textsc{Minimum Vertex Cover Problem} (MVCP) for $3$-regular graphs which is known to be APX-complete \cite{AlimontiKann2000}.

\begin{theorem} \label{th5.2}
The \textsc{Minimum Disjunctive Domination Problem} is APX-complete  for bipartite graphs with maximum degree $3$.
\end{theorem}
\begin{proof}
To show the APX-completeness of MDDP, it is enough to construct an L-reduction $f$ from the instances of MVCP to the instances of MDDP. Given a graph $G=(V,E)$, where $V=\{v_{1},v_{2},\ldots,v_{n}\}$, and $E=\{e_{1},e_{2},\ldots,e_{m}\}$, we construct a graph $H=(V_{H},E_{H})$ by replacing each edge $e_{i}=v_{r}v_{s}$ with the gadget $H_i$ as shown in Figure~\ref{fig:1}. Clearly, $H$ is a bipartite graph and maximum degree of $H$ is $3$.

\begin{figure}[h!]
  \begin{center}
  \includegraphics[width=5.5cm, height=4.8cm]{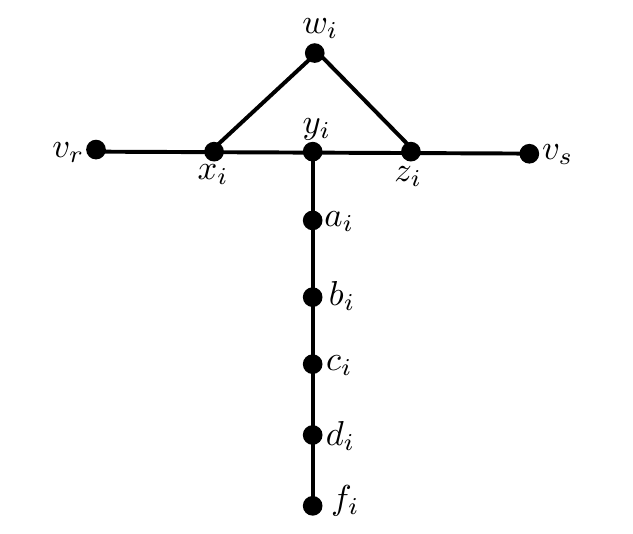}\\
 \caption{Graph $H_{i}$}
\label{fig:1}
\end{center}
\end{figure}

Now, we first prove the following claim:
\begin{claim}\label{cl3}
Let $D_{d}$ be a disjunctive dominating set of $H$ of cardinality at most $k$. Then, there exists a disjunctive dominating set, say $D'_{d}$, of $H$ of cardinality at most $k$ such that $\{y_{i},d_{i}\mid 1\leq i \leq m\}\subseteq D'_{d}$ for each $i\in \{1,2,\ldots, m\}$. In addition, for each edge $e_{i}$ in graph $G$, at least one of the end point of $e_{i}$ is present in $D'_{d}$.
\end{claim}
\begin{proof}(Proof of the claim)
For some $i$, if $d_{i}\notin D_{d}$, then $f_{i}$ must belong to $D_{d}$. In that case, if we remove $f_{i}$ from the set $D_{d}$ and add $d_{i}$ in the set $D_{d}$, then $D_{d}$ still remains a disjunctive dominating set of $H$.

So, we assume that $\{d_{i}\mid 1\leq i \leq m\}\subseteq D_{d}$. Now, to disjunctively dominate the vertex $b_{i}$, at least one vertex from the set $\{c_{i},b_{i},a_{i},y_{i}\}$ must belong to $D_{d}$. If $y_{i}\notin D_{d}$, then remove a vertex from the set $\{c_{i},b_{i},a_{i}\}\cap D_{d}$ from $D_{d}$ and add $y_{i}$ in $D_{d}$. Clearly, $D_{d}$ still remains a disjunctive dominating set of $H$ of same cardinality. Hence, given a disjunctive dominating set, say $D_{d}$, we can always construct a disjunctive dominating set, say $D'_{d}$, such that $\{d_{i},y_{i}\mid 1\leq i \leq m\}\subseteq D'_{d}$ and $|D_{d}'|\leq |D_{d}|$.

Now, we start with a disjunctive dominating set, say $D_{d}$, such that $\{d_{i},y_{i}\mid 1\leq i \leq m\}\subseteq D_{d}$.
Let $S=\{y_{i},d_{i}\mid 1\leq i \leq m\}$, and  $v_{r},v_{s}$ are end points of edge $e_{i}$ in graph $G$. The set $S$ disjunctively dominates all the vertices of $H$ except the $w'_{i}s$. Also, for each $w_{i}$, $S$ contains a vertex which is at distance two from $w_{i}$. Then, to disjunctively dominate the vertex $w_{i}$ in graph $H$, at least one vertex from the set $\{w_{i},x_{i},z_{i},v_{r},v_{s}\}$ must belong to $D_{d}$. Now, if $w_{i},x_{i}$ or $z_{i}$ belong to $D_{d}$, then remove them from $D_{d}$, and add either $v_{r}$ or $v_{s}$ in $D_{d}$. The resultant set $D_{d}$ still remains a disjunctive dominating set of $H$ of same or less cardinality. Note that, for each edge $e_{i}$, one of the endpoint is contained in the modified set $D_{d}$. This completes the proof of the claim.
\end{proof}

\begin{claim} \label{cl4}
$G$ has a vertex cover of cardinality at most $k$ if and only if $H$ has a disjunctive dominating set of cardinality at most $k+2m$.
\end{claim}
\begin{proof}(Proof of the claim)
Let $C$ be a vertex cover of $G$ of cardinality at most $k$. Then, it can be easily verified that $D_d= C\cup \{y_{i},d_{i} \mid 1\leq i \leq m\}$ is a disjunctive dominating set of $H$ of cardinality $k+2m$.

Conversely, suppose that $D_{d}$ is a disjunctive dominating set of $H$ of cardinality at most $k+2m$. Then by Claim~\ref{cl3}, we may assume that $\{y_{i},d_{i}\mid 1\leq i \leq m\}\subseteq D_{d}$, and for each edge $e_{i}$ in graph $G$, at least one of the end point of $e_{i}$ is contained in the set $D_{d}$. Thus, $D_{d}\cap V$ is a vertex cover of $G$ of cardinality at most $k$.
\end{proof}

From Claim~\ref{cl3} and Claim~\ref{cl4}, any disjunctive dominating set $D_d$ of $H$ can be transformed into a vertex cover $C$ of $G$ of cardinality at most $|D_d|-2m$. Let $D^*_d$ be a minimum disjunctive dominating set of $H$ and $C^*$ be a minimum vertex cover of $G$, then $|C^*|=|D^*_d|-2m$. Hence, we have $||C|-|C^{*}||\leq ||D_d|-|D_d^{*}||$. On the other hand, since $G$ is a 3-regular graph, $m\leq 3|V_{c}^{*}|$. Hence $|D_{d}^{*}|=|V_{c}^{*}|+2m\leq 7|V_{c}^{*}|$.

Hence $f$ is an L-reduction with $\alpha=7$ and $\beta=1$.
\end{proof}

%\label{claim:vctodds}

%Due to the construction used in Theorem~\ref{th5.2}, it can be seen that the constructed graph $H$ is bipartite. Hence we have the following corollary.
%\begin{corollary}
%The \textsc{Minimum Disjunctive Domination} problem is APX-complete for
%bipartite graphs with maximum degree $3$.
%\end{corollary}

\section{Conclusion}\label{sec:conclu}

In this article, we have proposed a linear time algorithm for MDDP in proper interval graphs. We have also tightened the NP-completeness of DDDP by showing that it remains NP-complete even in chordal graphs. From approximation point of view, we have proposed an approximation algorithm for MDDP in general graphs and have shown that this problem is APX-complete for bipartite graphs with maximum degree $3$. Note that, the results presented in this paper, can easily be extended to $b$-disjunctive dominating set for $b\geq 3$. It would be interesting to study the complexity of this problem in other graph classes and also the relation between disjunctive domination number and other domination parameters.


\begin{thebibliography}{}


\bibitem{AlimontiKann2000} P. Alimonti and V. Kann,
\newblock Some APX-completeness results for cubic graphs,
\newblock \emph{Theoret Comput Sci} 237(1-2) (2000) 123--134.

\bibitem{ambalath} A. M. Ambalath, R. Balasundaram, C. R. H., V. Koppula, N. Misra, G. Philip, and M. S. Ramanujan,
\newblock On the kernelization complexity of colorful motifs,
\newblock in {\em IPEC, Lecture Notes in Computer Science}, 6478 (2010) 14-–25.

\bibitem{Ausiello1999} G. Ausiello, P. Crescenzi, G. Gambosi, V. Kann, A. Marchetti-Spaccamela, and M. Protasi,
\newblock Complexity and approximation,
\newblock {\em Springer}, Berlin, (1999).


\bibitem{Bertossi} A. A. Bertossi,
\newblock Dominating sets for split and bipartite graphs,
\newblock {\em Inf. Process. Lett.}, 19(1) (1984) 37--40.

\bibitem{chlebik} M. Chleb\'{\i}k and J. Chleb\'{\i}kov\'{a}.
\newblock {Approximation hardness of dominating set problems in bounded degree graphs},
\newblock{\em Inform. and Comput.}, 206 (2008) 1264--1275.

\bibitem{dankelmann} P. Dankelmann, D. Day, D. Erwin, S. Mukwembi, and H. Swart,
\newblock Domination with exponential decay,
\newblock {\em Discrete Math.}, 309 (2009) 5877--5883.


\bibitem{fulkerson} D. R. Fulkerson and O. A. Gross,
\newblock Incidence matrices and interval graphs,
\newblock {\em Pacific J. Math.}, 15 (1965) 835--855.

\bibitem{goddard} W. Goddard, M. A. Henning, and C. A. McPillan,
\newblock The disjunctive domination number of a graph,
\newblock {\em Quaestiones Math.} 37(4) (2014) 547--561.



\bibitem{Haynes1} T. W. Haynes, S. T. Hedetniemi, and P. J. Slater,
\newblock Fundamentals of domination in graphs,
\newblock{\em Marcel Dekker Inc.}, New York, (1998).

\bibitem{Haynes2} T. W. Haynes, S. T. Hedetniemi, and P. J. Slater,
\newblock Domination in Graphs, Advanced Topics,
\newblock{\em Marcel Dekker Inc.}, New York, (1998).

%\bibitem{Hedetniemi} S.M. Hedetniemi, S.T. Hedetniemi, J. Knisely, and D.F. Rall.
%\newblock Secondary domination in graphs,
%\newblock {\em AKCE Int. J. Graphs Comb.}, 5 (2008) 103–-115.
%

\bibitem{Henning2014} M. A. Henning and S. A. Marcon,
\newblock Domination versus disjunctive domination in trees,
\newblock {\em Discrete Appl. Math.}, (2014).

\bibitem{total_disj} M. A. Henning and V. Naicker,
\newblock Disjunctive total domination in graphs,
\newblock {\em J. Comb. Optim.}, (2014) doi:10.1007/s10878-014-9811-4.

\bibitem{jacobson} M. S. Jacobson and K. Peters,
\newblock Complexity questions for $n$-domination and related parameters,
\newblock {\em Eighteenth Manitoba Conference on Numerical Mathematics and Computing (Winnipeg, MB,1988)} Congr. Numer. 68 (1989) 7–22.


\bibitem{jamison} R. E. Jamison and R. Laskar,
\newblock Elimination orderings of chordal graphs,
\newblock In {\em Combinatorics and applications ({C}alcutta, 1982)}, {ISI}, Calcutta, (1984),  192--200.

\bibitem{panda} B. S. Panda and S. K. Das,
\newblock A linear time recognition algorithm for proper interval
graphs, \newblock {\em Inform. Process. Lett.}, 87(3) (2003) 153--161.

\bibitem{rajgopalan} S. Rajgopalan and V. V. Vazirani,
\newblock Primal-dual RNC approximation algorithms for set cover and covering integer programs,
\newblock {\em SIAM J. Comput.}, 28 (1999) 526--541
\end{thebibliography}
\end{document}